\documentclass[a4paper,onecolumn,superscriptaddress,11pt,%
				unpublished,%
				]{quantumarticle}
\pdfoutput=1

\usepackage[utf8]{inputenc}
\usepackage[english]{babel}
\usepackage{amsthm,amsmath}
\usepackage{hyperref}
\usepackage[numbers,sort&compress]{natbib}
\usepackage[all]{hypcap}
\usepackage{stmaryrd}
\usepackage{graphicx}
\usepackage{keycommand}
\usepackage{multicol}
\usepackage{array}
\usepackage[normalem]{ulem}

\theoremstyle{definition}
\newtheorem{theorem}{Theorem}

\newtheorem{proposition}[theorem]{Proposition}

\newtheorem{definition}[theorem]{Definition}

\newtheorem{example*}[theorem]{Example*}
\newtheorem{examples*}[theorem]{Examples*}

\newtheorem{remark*}[theorem]{Remark*}

\newtheorem*{theorem*}{Theorem}
\newtheorem*{corollary*}{Corollary}
\newtheorem*{lemma*}{Lemma}
\newtheorem*{proposition*}{Proposition}

\usepackage{tikzit}
\input{zh.tikzdefs}

\tikzstyle{dot}=[inner sep=0.3mm, minimum width=2mm, minimum height=2mm, draw, shape=circle, font={\footnotesize}, tikzit fill=magenta]
\tikzstyle{hadamard}=[fill=zxhad, draw, inner sep=0.6mm, minimum height=1.5mm, minimum width=1.5mm, shape=rectangle, tikzit shape=rectangle, tikzit category=ZH-pf, tikzit fill=yellow]
\tikzstyle{small hadamard}=[hadamard]
\tikzstyle{lambda}=[hadamard, fill={rgb,255: red,180; green,180; blue,180}, tikzit shape=rectangle]
\tikzstyle{halfscalar}=[star, fill=black, draw=black, minimum size=8pt, inner sep=0pt]
\tikzstyle{box}=[shape=rectangle, text height=1.5ex, text depth=0.25ex, yshift=0.2mm, fill=white, draw=black, minimum height=3mm, minimum width=5mm, font={\small}]
\tikzstyle{Z dot}=[inner sep=0mm, minimum size=2mm, shape=circle, draw=black, fill={zx_green}, tikzit fill=green]
\tikzstyle{Z phase dot}=[minimum size=5mm, font={\footnotesize\boldmath}, shape=rectangle, rounded corners=2mm, inner sep=0.2mm, outer sep=-2mm, scale=0.8, tikzit shape=circle, draw=black, fill={zx_green}, tikzit draw=blue, tikzit fill=green]
\tikzstyle{X dot}=[Z dot, shape=circle, draw=black, fill={zx_red}, tikzit fill=red]
\tikzstyle{X phase dot}=[Z phase dot, tikzit shape=circle, tikzit draw=blue, fill={zx_red}, font={\footnotesize\color{black}\boldmath}, tikzit fill=red]
\tikzstyle{H box}=[hadamard]
\tikzstyle{st}=[star, star points=5, fill=white, draw=black, inner sep=1.2pt, line width=1.2pt, tikzit fill=blue, tikzit draw=red, tikzit category=ZH-pf]
\tikzstyle{triangle}=[regular polygon, regular polygon sides=3, fill=white, draw=black, inner sep=0pt, minimum width=1em, tikzit draw=blue, tikzit category=ZH-pf, tikzit fill=cyan]
\tikzstyle{not}=[fill={rgb,255: red,180; green,180; blue,180}, draw=black, shape=circle, font={$\neg$}, dot]
\tikzstyle{vertex}=[inner sep=0mm, minimum size=1mm, shape=circle, draw=black, fill=black]
\tikzstyle{vertex set}=[inner sep=0mm, minimum size=1mm, shape=circle, draw=black, fill=white, font={\footnotesize\boldmath}]
\tikzstyle{wide point}=[fill=white, draw, shape=isosceles triangle, shape border rotate=-90, isosceles triangle stretches=true, inner sep=0pt, minimum width=1.5cm, minimum height=6.12mm, yshift=-0.0mm]
\tikzstyle{medium gray box}=[semilarge box, fill={rgb,255: red,180; green,180; blue,180}]
\tikzstyle{small box}=[rectangle, fill=white, draw, minimum height=5mm, yshift=-0.5mm, minimum width=5mm, font={\small}]
\tikzstyle{small gray box}=[small box, fill={rgb,255: red,180; green,180; blue,180}]
\tikzstyle{medium box}=[rectangle, fill=white, draw, minimum height=10mm, yshift=-0.5mm, minimum width=8mm, font={\small}]
\tikzstyle{ddot}=[line width=1.6pt, inner sep=0mm, minimum width=2.5mm, minimum height=2.5mm, draw, shape=circle]
\tikzstyle{dd white}=[ddot, fill=white, tikzit draw=green]
\tikzstyle{dd white phase}=[white phase dot, line width=1.6pt, tikzit draw=yellow]
\tikzstyle{dd gray}=[ddot, fill={rgb,255: red,180; green,180; blue,180}, tikzit draw=green]
\tikzstyle{dd gray phase}=[gray phase dot, line width=1.6pt, tikzit draw=yellow]
\tikzstyle{empty diagram}=[draw={gray!40!white}, dashed, shape=rectangle, minimum width=1cm, minimum height=1cm]
\tikzstyle{empty diagram small}=[draw={gray!50!white}, dashed, shape=rectangle, minimum width=0.6cm, minimum height=0.5cm]

\tikzstyle{white dot}=[Z dot]
\tikzstyle{white phase dot}=[Z phase dot, tikzit shape=circle, tikzit fill=white, tikzit draw=blue]
\tikzstyle{gray dot}=[X dot, tikzit fill={rgb,255: red,180; green,180; blue,180}]
\tikzstyle{gray phase dot}=[X phase dot, tikzit shape=circle, tikzit draw=blue]

\tikzstyle{simple}=[-]
\tikzstyle{hadamard edge}=[-, dashed, dash pattern=on 2pt off 1pt, thick, draw=blue]
\tikzstyle{gray}=[-, draw={blue!60!white}, tikzit draw=blue]
\tikzstyle{blue}=[-, draw={blue!60!white}, tikzit draw=blue]
\tikzstyle{brace edge}=[-, tikzit draw=blue, decorate, decoration={brace,amplitude=1mm,raise=-1mm}]
\tikzstyle{diredge}=[->]
\tikzstyle{not edge}=[-, dashed, dash pattern=on 2pt off 1.5pt, thick, draw={rgb,255: red,255; green,68; blue,68}]
\tikzstyle{double edge}=[-, double, shorten <=-1mm, shorten >=-1mm, double distance=2pt]
\tikzstyle{boldedge}=[-, line width=1.6pt, shorten <=-0.17mm, shorten >=-0.17mm, tikzit draw=blue]

\usepackage[all]{hypcap} 

\makeatletter
\newcommand\etc{etc\@ifnextchar.{}{.\@}\xspace}

\makeatother

\usepackage{bm}

\newcommand{\norm}[1]{\ensuremath{\lVert#1\rVert}}

\newcommand{\R}{\mathbb{R}}

\begin{document}
\title{Optimising quantum circuits is generally hard}
\date{August 12th 2024}

\author{John van de Wetering}
\email{john@vdwetering.name}
\homepage{http://vdwetering.name}
\affiliation{University of Amsterdam}

\author{Matthew Amy}
\email{meamy@sfu.ca}
\homepage{https://www.cs.sfu.ca/~meamy/}
\affiliation{Simon Fraser University}

\begin{abstract}
	In order for quantum computations to be done as efficiently as possible it is important to optimise the number of gates used in the underlying quantum circuits. In this paper we find that many gate optimisation problems for approximately universal quantum circuits are NP-hard. In particular, we show that optimising the T-count or T-depth in Clifford+T circuits, which are important metrics for the computational cost of executing fault-tolerant quantum computations, is NP-hard by reducing the problem to Boolean satisfiability. With a similar argument we show that optimising the number of CNOT gates or Hadamard gates in a Clifford+T circuit is also NP-hard. Again varying the same argument we also establish the hardness of optimising the number of Toffoli gates in a reversible classical circuit. We find an upper bound to the problems of T-count and Toffoli-count of $\text{NP}^{\text{NQP}}$.
	Finally, we also show that for any non-Clifford gate $G$ it is NP-hard to optimise the $G$-count over the Clifford+$G$ gate set, where we only have to match the target unitary within some small distance in the operator norm.
\end{abstract} 

\maketitle

\section{Introduction}

Applications of quantum computers can roughly be divided into two domains: those which work on relatively small and noisy devices, and those which work only on large-scale fault-tolerant machines.
In the first category exist for instance a variety of variational algorithms that combine classical optimization with small, parameterized quantum circuits~\cite{choi2019tutorial,tilly2022variational}. For these circuits, noise accumulates over time and after every operation. It is then crucial to optimise the \emph{depth} of the circuit, corresponding to the runtime of the computation, as well as the number of multi-qubit operations (like CNOT gates), as those gates tend to be much more noisy than single-qubit operations.

In the second category we find most of the proven speedups of quantum computers, like applications of Grover's~\cite{grover1996fast} or Shor's algorithm~\cite{shor1994algorithms} or solving quantum chemistry problems~\cite{bauer2020quantum}. For such large-scale quantum computations we require a way to combat the accumulation of errors. This can be done by using \emph{quantum error correction}~\cite{shor1995scheme}. This allows us to distribute the logical quantum information over multiple physical qubits, which we can combine with clever measurements in order to detect and correct potential errors. If we want to do operations on the encoded logical information, we however need to implement this in such a way on the physical hardware that it does not spread errors in an uncorrectable way. That is, we need to implement the logical operations in a \emph{fault-tolerant} manner. This puts restrictions on the types of gates we can use, and it makes certain gates much more expensive than they would be when implemented on a physical level~\cite{fowler2012surface}. In particular, in many fault-tolerant architectures \emph{non-Clifford} gates are much more expensive to implement than \emph{Clifford} gates. A common choice of non-Clifford gate to make the gate set universal is the \emph{T gate} $\text{diag}(1,e^{i\frac\pi4})$. While in recent works ways to implement the $T$ gate (via magic state distillation) have been improved significantly, their cost still encompasses a large part of all resources required for a computation~\cite{Litinski2019gameofsurfacecodes,Gidney2019efficientmagicstate,Gidney2021howtofactorbit,Litinski2019magicstate}.

Simplifying the story somewhat, we can hence say that optimising a fault-tolerant quantum circuit means optimising the number of T gates, known as the T-count, required to implement the circuit.
There have hence been a multitude of results, both heuristics and optimal algorithms, for optimising the T-count of a given circuit or unitary.
These can be divided into three classes. First, there are the methods that treat all non-Clifford phase rotations equivalently, and just fuse together phases that can be brought together using the sum-over-paths method or by representing the gates by a series of Pauli exponentials~\cite{amy2014polynomial,kissinger2019tcount,zhang2019optimizing,vandewetering2024optimal}. These techniques are efficient, but generally don't find the optimal T-count of a circuit. Second, there are the methods that synthesise directly from the matrix, and do not optimise a circuit in place~\cite{gheorghiu2022t,gheorghiu2022quasi}. These produce optimal T-counts by design, but are infeasible to run on circuits beyond just a couple of qubits in size.
Third, there are the methods that use the equivalence between optimising the T-count of diagonal CNOT+T circuits and well-studied problems like symmetric 3-tensor factorisation and Reed-Muller decoding~\cite{amy2016t,deBeaudrapN2020treducspidernest,heyfron2018efficient,ruiz2024quantum}. 

While more powerful than the phase-fusing techniques and more efficient than the direct synthesis techniques, the methods in the third group still grow prohibitively costly for large circuits (beyond roughly 50 qubits). This is not too surprising as Reed-Muller decoding and symmetric 3-tensor factorisation are believed to be hard problems. However, to our knowledge no concrete hardness result is known for optimising the T-count of a general Clifford+T circuit. In this paper we demonstrate with a simple argument that this problem is at least NP-hard. This argument turns out to be easily adaptable to prove the hardness of optimising several other types of gates: Toffolis, CNOTs, Hadamards, and in fact any other non-Clifford gate.

\subsection{Statement of results}

We will define the problem T-COUNT as follows: given an integer $k$ and a Clifford+T circuit implementing a unitary $U$, determine whether there exists a Clifford+$T$ circuit implementing $U$ using at most $k$ T gates. Note that the optimisation version of the problem reduces to the T-COUNT problem via a binary search running logarithmically in the length of the circuit. We can then state our main result.
\begin{theorem}
T-COUNT is NP-hard under polynomial-time Turing reductions.
\end{theorem}

Most T gates in many applications are used inside of Toffoli gates $\text{Tof}\ket{x,y,z}\mapsto \ket{x,y,z\oplus (xy)}$, as part of the synthesis of classical functions being applied to quantum states (such as the synthesis of the modular exponentiation needed in Shor's algorithm). An interesting related question is hence to optimise the number of Toffoli gates in such a classical circuit consisting of Toffoli, CNOT and NOT gates. Defining the problem of TOF-COUNT similarly to T-COUNT, we show how a similar argument for the hardness of T-COUNT can also be used to prove the hardness of TOF-COUNT.
\begin{theorem}
TOF-COUNT is NP-hard under polynomial-time Turing reductions.
\end{theorem}

We further show that optimisation of \emph{any} non-Clifford gate is NP-hard. Let $G$ be some non-Clifford gate. For some $\varepsilon>0$ we define the problem of $G$-COUNT$_\varepsilon$ as follows: Given an integer $k$ and a Clifford+$G$ circuit implementing some unitary $U$, determine whether there exists a Clifford+$G$ circuit containing at most $k$ $G$ gates which implements a unitary $U'$ that is within distance $\varepsilon$ of $U$ in the operator norm.
\begin{theorem}
	Let $G$ be any non-Clifford gate and $\varepsilon<\sin(\frac{\pi}{16})\approx 0.195$. Then $G$-COUNT$_\varepsilon$ is NP-hard.
\end{theorem}

We also find an upper bound for the T-COUNT and TOF-COUNT problems. To state this we require the complexity class NQP, standing for \emph{non-deterministic quantum polynomial} time. This class has as a complete problem determining whether two poly-size quantum circuits are \emph{exactly} equal (and hence should be contrasted with the more well-known complexity class QMA which has as a complete problem determining whether two circuits are \emph{approximately} equal).
\begin{theorem}
	The T-COUNT and TOF-COUNT problems are contained in $\text{NP}^{\text{NQP}}$.
\end{theorem}

Although non-Clifford gates are the majority of the cost of fault-tolerant implementations, other gates aren't free, and hence we would also like to optimise those. In particular, CNOT gates introduce connectivity constraints that might require expensive routing across distant qubits. Let CNOT-COUNT be defined analogously to T-COUNT, but with respect to CNOT gates (i.e.~count the number of CNOT gates in a Clifford+T circuit).
\begin{theorem}
CNOT-COUNT is NP-hard under polynomial-time Turing reductions.
\end{theorem}
The argument in fact works for any entangling gate in the Clifford+T gate set, and also works if we let CNOT gates between different qubit pairs carry different weights (as long as all those weights are non-zero).

We could express the Clifford+T gate set succinctly as consisting of CNOT, T and Hadamard gates. We have now seen that optimising the first two types of gates is NP-hard. Optimising Hadamard gates is in fact also hard.
\begin{theorem}
Hadamard-COUNT is NP-hard under polynomial-time Turing reductions.
\end{theorem}

Our $\text{NP}^{\text{NQP}}$ upper bound also applies to Hadamard-COUNT, but the argument does not extent to CNOT-COUNT. We currently do not know of an upper bound to that problem.

\section{Hardness of circuit optimisation}

We establish NP-hardness by reduction from Boolean satisfiability. Let $f:\{0,1\}^n \to \{0,1\}$ be some Boolean function, given as a Boolean expression. Using standard techniques we can build the classical oracle $U_f$ implementing $U_f\ket{\vec x, y} = \ket{\vec x, y\oplus f(\vec x)}$. Note that $U_f$ is an $(n+1)$-qubit quantum circuit which can be constructed as a poly$(n)$ size Clifford+T circuit (which requires potentially one borrowed ancilla). Consider then the following quantum circuit $C_f$:
\begin{equation}\label{eq:SAT-oracle}
  \begin{tikzpicture}
	\begin{pgfonlayer}{nodelayer}
		\node [style=box, minimum height=1.0cm] (0) at (4.5, 0.875) {$U_f$};
		\node [style=none] (1) at (1.5, 1.5) {};
		\node [style=none] (2) at (9.25, 1.5) {};
		\node [style=none] (3) at (1.5, 0.25) {};
		\node [style=none] (4) at (9.25, 0.25) {};
		\node [style=none] (5) at (1.5, -1) {};
		\node [style=none] (6) at (9.25, -1) {};
		\node [style=vertex set] (7) at (4.5, -1) {$\!\!+\!\!$};
		\node [style=none] (8) at (3.25, 1.125) {$\vdots$};
		\node [style=none] (9) at (8.75, 1.125) {$\vdots$};
		\node [style=box] (10) at (2.5, -1) {$T^\dagger$};
		\node [style=box] (11) at (6, -1) {$T$};
		\node [style=box, minimum height=1.0cm] (12) at (7.5, 0.875) {$U_f$};
		\node [style=vertex set] (13) at (7.5, -1) {$\!\!+\!\!$};
		\node [style=box, minimum height=1.6cm] (14) at (-3, 0.25) {$C_f$};
		\node [style=none] (15) at (-4.5, 1.5) {};
		\node [style=none] (16) at (-4.5, 0.25) {};
		\node [style=none] (17) at (-4.5, -1) {};
		\node [style=none] (18) at (-1.5, 1.5) {};
		\node [style=none] (19) at (-1.5, 0.25) {};
		\node [style=none] (20) at (-1.5, -1) {};
		\node [style=none] (21) at (-4, 1.125) {$\vdots$};
		\node [style=none] (22) at (-2, 1.125) {$\vdots$};
		\node [style=none] (23) at (0, 0) {=};
	\end{pgfonlayer}
	\begin{pgfonlayer}{edgelayer}
		\draw (1.center) to (2.center);
		\draw (4.center) to (3.center);
		\draw (6.center) to (5.center);
		\draw (0) to (7);
		\draw (12) to (13);
		\draw (17.center) to (20.center);
		\draw (19.center) to (16.center);
		\draw (15.center) to (18.center);
	\end{pgfonlayer}
\end{tikzpicture}
\end{equation}
It is straightforward to verify that $C_f$ implements the diagonal operation 
$$C_f\ket{\vec x,y} \ =\  e^{i\frac\pi4(1-2y)f(\vec x)} \ket{\vec x, y}.$$
Now, if $f$ is not satisfiable, then $f(\vec x) = 0$ for all $\vec x$, and hence we see that $C_f = \text{id}$. Additionally, if $f$ is satisfiable for all $\vec x$, then we have
$$C_f\ket{\vec x,y} = e^{i\frac\pi4 (1 - 2y)}\ket{\vec x,y}\ =\ e^{i\frac\pi4} e^{-i\frac\pi2 y} \ket{\vec x,y} \ = \ e^{i\frac\pi4} (I_n\otimes S^\dagger)\ket{\vec x,y}.$$
So, up to global phase, $C_f$ is just an $S^\dagger$ gate in this case, and hence Clifford. In either case $C_f$ is Clifford, so that the minimal T-count of $C_f$ is zero. 

Now suppose $f$ is satisfiable, but that not every input is a solution. Then there exist $\vec z_1$ and $\vec z_2$ such that $f(\vec z_1) = 1$ and $f(\vec z_2) = 0$. Then it is easy to see that
$$C_f\ket{\vec z_1, 0} \ =\  e^{i\frac\pi4} \ket{\vec z_1,0} \qquad \text{and} \qquad C_f\ket{\vec z_2, 0} = \ket{\vec z_2, 0}.$$

We can now observe that $C_f$ is non-Clifford by considering the action of $C_f$ on the $n$-qubit Pauli $X^{\vec z_1 \oplus \vec z_2}:= X^{(\vec z_1 \oplus \vec z_2)_1}\otimes \cdots \otimes X^{(\vec z_1 \oplus \vec z_2)_n}$. In particular,
$$
C_f^\dagger X^{\vec z_1 \oplus \vec z_2} C_f \ket{\vec z_1, 0} = e^{i\frac\pi4}C_f^\dagger X^{\vec z_1 \oplus \vec z_2}\ket{\vec z_1, 0} = e^{i\frac\pi4}C_f^\dagger \ket{\vec z_2, 0} = e^{i\frac\pi4}\ket{\vec z_2, 0},$$
and hence $C_f^\dagger X^{\vec z_1 \oplus \vec z_2} C_f$ is not a member of the $n$-qubit Pauli group. By definition $C_f$ is non-Clifford and so its minimal T-count over Clifford+$T$ is necessarily greater than 0.

To complete the reduction, given a Boolean expression $f$ build $C_f$ as above in poly time and determine whether a $T$-count 0 implementation exists. If the minimal $T$-count is greater than 0, $f$ is non-constant and hence satisfiable. If instead the minimal $T$-count is $0$, then either $f$ is not satisfiable or it is always satisfiable. We can distinguish between these two cases by evaluating $f(0\cdots 0)$. If $f(0 \cdots 0) = 1$, then $f$ is satisfiable, and otherwise we conclude that it must not be satisfiable.

Note that the exact value of $T=Z(\frac\pi4)$ is not that special. The argument continues to hold for any $Z(\alpha)$, as long as the resulting Clifford+$Z(\alpha)$ gate set allows you to construct the $U_f$ classical oracle. Hence, this argument also shows hardness of optimising the number of $Z(\frac{\pi}{2^n})$ for $n\geq 2$. In Section~\ref{sec:general-hardness} we show how to modify the argument to prove hardness of optimising \emph{any} non-Clifford gate.

In addition, we only had to distinguish here between a T-count of zero, and a non-zero T-count. This means that similar arguments would also hold for modified cost functions. We could for instance consider the decision problem of T-DEPTH, that asks whether a given circuit has an implementation that requires at most $k$ layers of parallel T gates. Of course any Clifford circuit will have a T-depth of zero, while any non-Clifford circuit will have a T-depth of at least one. Hence, the same argument as above shows that T-DEPTH is also NP-hard.
We can also consider the problem of IS-CLIFFORD, which asks whether the given Clifford+$T$ circuit implements a Clifford unitary, i.e.~whether it's T-count is zero. We then also see that IS-CLIFFORD is NP-hard.

\subsection{Hardness of Toffoli-count optimisation}

We can use a similar argument to the one above to show that the problem of TOFFOLI-COUNT, determining the minimal number of Toffoli gates needed to write down a classical reversible circuit (i.e.~a quantum circuit consisting of NOT, CNOT and Toffoli gates), is also NP-hard. We then replace the $C_f$ of Eq.~\eqref{eq:SAT-oracle} by the following:
\begin{equation}\label{eq:SAT-oracle-tof}
  \begin{tikzpicture}
	\begin{pgfonlayer}{nodelayer}
		\node [style=box, minimum height=1.0cm] (0) at (4.5, 0.875) {$U_f$};
		\node [style=none] (1) at (1.5, 1.5) {};
		\node [style=none] (2) at (9.25, 1.5) {};
		\node [style=none] (3) at (1.5, 0.25) {};
		\node [style=none] (4) at (9.25, 0.25) {};
		\node [style=none] (5) at (1.5, -0.75) {};
		\node [style=none] (6) at (9.25, -0.75) {};
		\node [style=vertex set] (7) at (4.5, -0.75) {$\!\!+\!\!$};
		\node [style=none] (8) at (3.25, 1.125) {$\vdots$};
		\node [style=none] (9) at (8.75, 1.125) {$\vdots$};
		\node [style=box, minimum height=1.0cm] (12) at (7.5, 0.875) {$U_f$};
		\node [style=vertex set] (13) at (7.5, -0.75) {$\!\!+\!\!$};
		\node [style=box, minimum height=2.4cm] (14) at (-3, -0.625) {$C_f$};
		\node [style=none] (15) at (-4.5, 1.5) {};
		\node [style=none] (16) at (-4.5, 0.25) {};
		\node [style=none] (17) at (-4.5, -0.75) {};
		\node [style=none] (18) at (-1.5, 1.5) {};
		\node [style=none] (19) at (-1.5, 0.25) {};
		\node [style=none] (20) at (-1.5, -0.75) {};
		\node [style=none] (21) at (-4, 1.125) {$\vdots$};
		\node [style=none] (22) at (-2, 1.125) {$\vdots$};
		\node [style=none] (23) at (0, 0) {=};
		\node [style=none] (24) at (-4.5, -1.75) {};
		\node [style=none] (25) at (-1.5, -1.75) {};
		\node [style=none] (26) at (-4.5, -2.75) {};
		\node [style=none] (27) at (-1.5, -2.75) {};
		\node [style=vertex] (28) at (2.75, -0.75) {};
		\node [style=vertex] (29) at (6, -0.75) {};
		\node [style=none] (30) at (1.5, -1.75) {};
		\node [style=none] (31) at (9.25, -1.75) {};
		\node [style=none] (32) at (1.5, -2.75) {};
		\node [style=none] (33) at (9.25, -2.75) {};
		\node [style=vertex] (34) at (2.75, -1.75) {};
		\node [style=vertex] (35) at (6, -1.75) {};
		\node [style=vertex set] (36) at (2.75, -2.75) {$\!\!+\!\!$};
		\node [style=vertex set] (37) at (6, -2.75) {$\!\!+\!\!$};
	\end{pgfonlayer}
	\begin{pgfonlayer}{edgelayer}
		\draw (1.center) to (2.center);
		\draw (4.center) to (3.center);
		\draw (6.center) to (5.center);
		\draw (0) to (7);
		\draw (12) to (13);
		\draw (17.center) to (20.center);
		\draw (19.center) to (16.center);
		\draw (15.center) to (18.center);
		\draw (24.center) to (25.center);
		\draw (26.center) to (27.center);
		\draw (30.center) to (31.center);
		\draw (32.center) to (33.center);
		\draw (28) to (36);
		\draw (29) to (37);
	\end{pgfonlayer}
\end{tikzpicture}
\end{equation}
Again if $f$ is not satisfiable, $C_f$ implements the identity, and if it is always satisfiable, then it implements a CNOT on the bottom two qubits. In both cases the Toffoli-count is zero. Otherwise if $f(\vec z_1) = 1$ and $f(\vec z_2) = 0$ we can check that $C_f^\dagger X^{\vec z_1\oplus \vec z_2} C_f\ket{\vec z_1,x,y,z} = \ket{\vec z_2, x,y, y\oplus z}$, so that $C_f$ is not Clifford, and hence its Toffoli count is not zero.

As above, this also shows that the problem TOF-DEPTH is NP-hard. We can also consider the problem IS-LINEAR, which asks whether that classical reversible circuit implements a linear Boolean function (one which can be implemented solely using XOR and NOT operations), and we then see that this is also NP-hard.

\subsection{Hardness of CNOT-count optimisation}

We can reuse the argument for NP-hardness of T-COUNT to prove the hardness of optimising the number of entangling gates in a Clifford+$T$ circuit. In the case when $CNOT$ is the only multi-qubit gate over a basis of the Clifford+$T$ operators, as in the canonical generators $\{H, T, CNOT, S:=T^2\}$, this implies the hardness of $CNOT$-count optimization. Likewise, if $CZ$ is used in place of the $CNOT$ gate as the single entangling gate, this implies optimization of $CZ$-count is again NP-hard.

We define the problem of ENT-COUNT analogously to T-COUNT, where we replace the role of the $T$ gate by any multi-qubit (entangling) Clifford+$T$ gate $U$, for instance $CNOT$, $CZ$, or an entangling product of Clifford and $T$ gates. We require that $U$ is a Clifford+$T$ operator in order to guarantee that the circuit $C_f$ admits implementation over single-qubit Clifford+$T$ and $U$. We recall from above that for $C_f$ in Eq.~\eqref{eq:SAT-oracle}, if $f$ is not satisfiable, $C_f = \text{id}$, and hence it has ENT-COUNT zero. If instead $f$ is always satisfiable, $C_f = e^{i\frac\pi4}(I_n\otimes S^\dagger)$ so that its ENT-COUNT is also zero. Now suppose $f$ is non-constant and pick $\vec z_1$ and $\vec z_2$ such that $f(\vec z_1) = 1$ and $f(\vec z_2) = 0$. We can observe that:
\begin{align*}
	C_f\ket{\vec z_1, 0} \ &= \ e^{i\frac\pi4}\ket{\vec z_1,0} \\
	C_f\ket{\vec z_2, 0} \ &= \ \ket{\vec z_2,0}\\
	C_f\ket{\vec z_1, 1} \ &= \ e^{-i\frac\pi4}\ket{\vec z_2,1}\\
	C_f\ket{\vec z_2, 1} \ &= \ \ket{\vec z_2,1}
\end{align*}
It can be observed that the right-hand side above is non-separable, and in particular $C_f$ is entangling on this $4$-dimensional subspace. Since $C_f$ is itself entangling, it can't be written as a product of non-entangling operators, and so at least one $U$ gate is necessary. Hence, the ability to determine the optimal $U$ count of $C_f$ allows us to determine whether $f$ is satisfiable.

As with the situation for T-COUNT, the same argument can be used to argue that ENT-DEPTH is NP-hard. In fact, the same argument works for any cost model where zero entangling gates has a cost of zero, while any non-zero amount of entangling gates has a non-zero cost, and hence this argument also applies to a setting where we are only allowed to put entangling gates in certain locations, or different locations occur varying (non-zero) costs.

\subsection{Hardness of Hadamard-count optimization}

Modifying the $T$-count optimization hardness argument in yet another different way also suffices to prove hardness for the H-COUNT optimisation problem --- determining the minimal number of hadamard gates needed to implement a circuit over the canonical generators $\{H, T, CNOT, S\}$ of the Clifford+$T$ gate set. In particular, by conjugating the target bit with hadamard gates as below, it can be observed that $C_f$ can be implemented with $H$-count zero if and only if $f$ is unsatisfiable.

\begin{equation}\label{eq:SAT-oracle-H}
  \begin{tikzpicture}
	\begin{pgfonlayer}{nodelayer}
		\node [style=box, minimum height=1.0cm] (0) at (4.5, 0.875) {$U_f$};
		\node [style=none] (1) at (0, 1.5) {};
		\node [style=none] (2) at (10.5, 1.5) {};
		\node [style=none] (3) at (0, 0.25) {};
		\node [style=none] (4) at (10.5, 0.25) {};
		\node [style=none] (5) at (0, -1) {};
		\node [style=none] (6) at (10.5, -1) {};
		\node [style=vertex set] (7) at (4.5, -1) {$\!\!+\!\!$};
		\node [style=none] (8) at (3.25, 1.125) {$\vdots$};
		\node [style=none] (9) at (8.75, 1.125) {$\vdots$};
		\node [style=box] (10) at (2.5, -1) {$T^\dagger$};
		\node [style=box] (11) at (6, -1) {$T$};
		\node [style=box] (10) at (1, -1) {$H$};
		\node [style=box] (11) at (9, -1) {$H$};
		\node [style=box, minimum height=1.0cm] (12) at (7.5, 0.875) {$U_f$};
		\node [style=vertex set] (13) at (7.5, -1) {$\!\!+\!\!$};
		\node [style=box, minimum height=1.6cm] (14) at (-4.5, 0.25) {$C_f$};
		\node [style=none] (15) at (-6, 1.5) {};
		\node [style=none] (16) at (-6, 0.25) {};
		\node [style=none] (17) at (-6, -1) {};
		\node [style=none] (18) at (-3, 1.5) {};
		\node [style=none] (19) at (-3, 0.25) {};
		\node [style=none] (20) at (-3, -1) {};
		\node [style=none] (21) at (-5.5, 1.125) {$\vdots$};
		\node [style=none] (22) at (-3.5, 1.125) {$\vdots$};
		\node [style=none] (23) at (-1.5, 0) {=};
	\end{pgfonlayer}
	\begin{pgfonlayer}{edgelayer}
		\draw (1.center) to (2.center);
		\draw (4.center) to (3.center);
		\draw (6.center) to (5.center);
		\draw (0) to (7);
		\draw (12) to (13);
		\draw (17.center) to (20.center);
		\draw (19.center) to (16.center);
		\draw (15.center) to (18.center);
	\end{pgfonlayer}
\end{tikzpicture}
\end{equation}

In particular, if $f$ is unsatisfiable, then $C_f$ is the identity. If however $f$ is satisfiable, then there exists at least one bit string $\vec z$ such that $f(\vec z) = 1$, and so where the control register is in the state $\ket{\vec z}$, $C_f$ implements the following transformation on the target bit:
\[
	HXTXT^\dagger = \begin{bmatrix} \frac{1}{\sqrt{2}} & \frac{i}{\sqrt{2}} \\ \frac{i}{\sqrt{2}} & \frac{1}{\sqrt{2}} \end{bmatrix}
\]
which can not be implemented over $\{H, T, CNOT, S\}$ without at least one hadamard gate (and hence neither can its controlled version), as the hadamard-free circuits over the canonical basis correspond to generalized permutations (permutations of computational basis states together with added phases on these basis states). More generally, optimising the $H$-count is hard over a gate set $\mathcal{G}$ whenever it generates the Clifford+$T$ circuits, and $\mathcal{G}\setminus \{H\}$ contains only $Z$- and $X-$basis transformations.

The fact that the H-COUNT optimisation problem is NP-hard is not surprising given its close connection and applications to $T$-count optimisation. In particular, the number of distinct $n$-qubit unitaries over Clifford+$T$ with a fixed number of hadamard gates is finite \cite{amy2017} and the minimal number of $T$-gates in an $n$-qubit Clifford+$T$ circuit with $k$ hadamard gates is bounded above by both $O(k\cdot n)$ and $O((n+k)^2)$ \cite{amy2016t}. Optimisation of the $H$-count over Clifford+$T$ hence has significant application to the problem of $T$-count optimisation.

\section{Hardness of general approximate circuit optimisation}\label{sec:general-hardness}
We can adapt the above arguments to show hardness of optimising the number of any specific non-Clifford gate.

\begin{definition}
	Let $G$ be any unitary quantum gate. We define the decision problem $G$-COUNT$_\varepsilon$ as follows: given the inputs
	\begin{itemize}
		\item $C$, a circuit specified in the Clifford+$G$ gate set;
		\item $k$, an integer;
		\item $\varepsilon$; an error bound in $\mathbb{R}_{>0}$;
	\end{itemize}
	determine whether there exists a circuit $C'$ over the inverse-closed Clifford+$G$ gate set using at most $k$ $G$ or $G^\dagger$ gates, such that for some global phase $\alpha$ we have $\norm{C-e^{i\alpha}C'}_{\infty} \leq \epsilon$.
\end{definition}

Note that $G$-COUNT$_\varepsilon$ is asking whether there exists any circuit close to the target circuit with some $G$-count. It hence allows us to determine the exact optimal $G$-count of circuits close to the target. This is different then asking to approximate the $G$-count itself.

\begin{theorem}
	For any non-Clifford gate $G$, $G$-COUNT$_\varepsilon$ is NP-hard under polynomial-time Turing reductions for error bounds $\varepsilon<\sin(\frac{\pi}{16})\approx 0.195$.
\end{theorem}
\begin{proof}
	We modify the argument we used above: reducing from Boolean satisfiability, by asking about the approximate $G$-count of the circuit $C_f$ in Eq.~\eqref{eq:SAT-oracle}. Note that as a circuit over Clifford+$T$, $C_f$ might not be exactly expressible as a circuit over Clifford+$G$. Suppose first that this is the case --- for instance, if $G=\sqrt{T}$. 
	We show that if $f$ is non-constant, then at least one $G$ gate is required to approximate $C_f$ to distance $\epsilon$ over the Clifford+$G$ gate set. In particular, assume that $f(\vec z_1) = 1$ and $f(\vec z_2) = 0$ for some vectors $\vec z_1$ and $\vec z_2$, and let $U$ be a Clifford. We will find a lower bound on $\norm{C_f - e^{i\alpha} U}_\infty$. 

	First we note that any global phase multiple of a non-diagonal Clifford has distance at least $\frac12$ from $C_f$. In particular, assume $U$ is not diagonal and let $\ket{\vec x,y}$ be some computational basis state such that $U\ket{\vec x,y}\not\propto \ket{\vec x,y}$. Since $U\ket{\vec x,y}$ is a stabiliser state, $\norm{\ket{\vec{x},y} - U\ket{\vec{x},y}}_2 \geq \frac12$ \cite{gmc17}, and in particular since $C_f$ is diagonal,
\[
	\norm{C_f - e^{i\alpha} U}_\infty \geq \norm{C_f\ket{\vec x,y} - e^{i\alpha}U\ket{\vec x,y}}_2 \geq \frac12.
\]

	Now suppose $U$ is diagonal. Then $U$ has the form $U\ket{\vec x,y} = e^{i\phi(\vec x,y)} \ket{\vec x,y}$ where $\phi$ is some quadratic function of $\vec x$ and $y$ taking values in $\frac\pi2\mathbb{Z}$. Then for any $\vec x$ and $y$:
	\[ \norm{C_f - e^{i\alpha} U}_\infty \ \geq \ \norm{C_f\ket{\vec x,y} - e^{i\alpha} U\ket{\vec x,y}}_2 \ = \  \lvert e^{i\frac\pi4 (1-2y)f(\vec x)} - e^{i\alpha}e^{i\phi(\vec x,y)} \rvert.\]
	Plugging in $\vec x = \vec z_1$ and $\vec x = \vec z_2$, we then find that for $U$ to be an $\varepsilon$-approximation of $C_f$, we must at least satisfy:
	\[\lvert e^{i\frac\pi4}e^{-i\frac\pi2y} - e^{i\alpha}e^{i\phi(\vec z_1,y)} \rvert \ \leq \ \varepsilon \qquad \text{and} \qquad \lvert 1 - e^{i\alpha}e^{i\phi(\vec z_2,y)} \rvert \ \leq \ \varepsilon.\]
	Assuming without loss of generality that $0\leq \alpha\leq \frac\pi4$ (since if it is outside this bound, a $\frac\pi2$ phase can be extracted into the Clifford circuit itself), we see that to minimise the value of both of these terms, we need to have $\phi(\vec z_2,y) = 0$, $\phi(\vec z_1,y) = -\frac\pi2 y$ and $\alpha=\frac\pi8$. Any other value of $\phi$ would just increase the value of the expression, while changing $\alpha$ would decrease one at the cost of the other. Hence, the closest a global phase multiple of a Clifford can get to approximating $C_f$ is at least $\lvert 1 - e^{i\frac\pi8} \rvert  = 2\sin(\frac{\pi}{16})\approx 0.39$. 

	However, all of this is assuming that $C_f$ can be exactly implemented as a Clifford+$G$ circuit, so that it can be given as an input to $G$-COUNT$_\varepsilon$. This is not the case in general.
	Note though that if $G$ is non-Clifford that Clifford+$G$ will be approximately universal, and hence that it can approximate any unitary. In particular, due to the Solovay-Kitaev theorem it can approximate any other finite gate set with polylogarithmic overhead in the error. 
	We can then translate every $T$ gate in $C_f$ into a circuit over the Clifford+$G$ gate set, giving a circuit $C'_f$ that is within $\varepsilon$ of $C_f$. Since $G$ and $T$ are fixed, and there are a polynomial number of them, this takes polynomial time in $\log 1/\varepsilon$. Now, if $f$ is always satisfiable or not satisfiable, then $C_f$ is Clifford, so that $C'_f$ is within $\varepsilon$ to being a Clifford. Since we only have to estimate up to $\varepsilon$, a Clifford circuit will do, and hence the $G$-count is zero. 
	Conversely, suppose the $G$-count of $C'_f$ is zero, so that it is within $\varepsilon$ of some Clifford. Using the triangle inequality, $C_f$ must then be within $2\varepsilon$ of some Clifford. However, by assumption $2\varepsilon<2\sin(\frac{\pi}{16}) = \lvert 1 - e^{i\frac\pi8} \rvert$ so this is only possible if $C_f$ is Clifford itself by the argument above, so that $f$ must be not satisfiable or everywhere satisfiable. Hence, if the $G$-count of $C'_f$ is greater than zero, then $f$ is satisfiable.

	The full reduction is then as follows: for a Boolean formula, construct $C_f$, written in the Clifford+$T$ gate set. Using the Solovay-Kitaev algorithm approximate the $T$ gate using Clifford+$G$ closely enough so that we get an approximation of $C_f$ that is within $\varepsilon$ in the operator norm. Ask whether $G$-COUNT$_\varepsilon(C'_f,0)$ for some $\varepsilon < \sin(\frac{\pi}{16})$ is true. If it is, then the $G$-count is zero, and hence the circuit is well-approximated by a Clifford, which is only possible if $f$ is not satisfiable, or if $f$ is always satisfiable. Test $f(0\cdots 0)$ to see whether it is indeed always satisfiable. If not, then it is not satisfiable. If $G$-COUNT$_\varepsilon(C'_f,0)$ is false, then any approximation to $C'_f$ and hence to $C_f$ can only be a non-Clifford circuit. Hence, $f$ must be satisfiable.
\end{proof}

The only thing we needed for this argument to work, was for the cost function (in this case: the number of $G$ gates) to distinguish between Clifford unitaries, and non-Clifford unitaries. Hence, we can generalise the above problem to gate sets containing multiple non-Clifford gates $G_1,\ldots, G_m$, and any cost function 
$$f:\{\text{circuits over Clifford+}\{G_1,\ldots, G_m\}\,\}\to \R_{\geq 0}$$
as long as $f(C) = 0$ means $C$ implements a Clifford unitary, and $f(C)>0$ means $C$ implements a non-Clifford unitary. Hence, determining the optimal $G$-depth, instead of $G$-count, will also be NP-hard.

\section{An upper bound to the hardness of T-count and Toffoli optimisation}

The arguments above show that T-COUNT and its associated optimisation problem are at least NP-hard (and the same for TOFFOLI-COUNT). Let us also demonstrate a simple upper bound to the T-COUNT problem. 
\begin{proposition}
	The T-COUNT problem is contained in $\text{NP}^{\text{NQP}}$.
\end{proposition}
We first recall that determining whether two poly-size quantum circuits are exactly equal is a coNQP-complete problem~\cite{tanaka2010exact} (non-deterministic quantum polynomial time). Note that a QMA oracle~\cite{bookatz2012qmacomplete} is not enough since we care about exact equality. Now to determine whether a given $n$-qubit circuit $C$ has an implementation with at most $k$ T gates, we realise first that such a circuit can be made to have at most $O(n^2k)$ Clifford+$T$ gates: any pure Clifford circuit can be represented by a normal form consisting of $O(n^2)$ gates~\cite{aaronsongottesman2004}, and hence a general Clifford+$T$ circuit containing $k$ $T$ gates can be written as a series of Clifford normal forms followed by a $T$ gate, with this structure repeated $k$ times (for the specific case of Clifford+$T$ this can actually be improved to $O(nk+n^2)$ by the use of Pauli exponentials~\cite{Litinski2019gameofsurfacecodes}). Hence, we can non-deterministically choose any circuit with up to $k$ $T$ gates in non-deterministic poly-time, and then use an NQP oracle to determine whether this circuit is equal to $C$. Hence T-COUNT is in $\text{NP}^{\text{NQP}}$.

Note that classical Boolean circuit minimisation is complete for $\Sigma_2^P := \text{NP}^{\text{NP}}$ \cite{buchfuhrer2011}, so that the only difference with this bound is that we replace the coNP problem of determining whether Boolean circuits are equal, with the coNQP problem of doing the same for quantum circuits.

The argument above works for determining an upper bound of exact optimisation of Clifford+$G$ circuits for any non-Clifford $G$, which hence includes TOF-COUNT.
To establish the same upper bound for H-COUNT, we need a slightly different argument. There we can use the fact that for a fixed number of qubits $n$, there is only a finite number of Hadamard-free unitaries over $\{\text{CNOT},S,T\}$~\cite{amy2017,amy2016t}, and in particular we can build any such circuit with $O(n^3)$ gates, which are easy to enumerate (a $O(n^2)$ circuit suffices~\cite{amy2016t}, but allowing $O(n^3)$ gates makes them easier to enumerate, albeit with some redundancy). We can then non-deterministically enumerate all $O(kn^3)$ Clifford+T circuits that have at most $k$ Hadamard gates. Hence, H-COUNT is in $\text{NP}^{\text{NQP}}$.

We cannot use a similar argument for bounding the hardness of CNOT-count, since there are infinitely many single-qubit unitaries over $\{H,T\}$, so it is not clear how we could enumerate all candidate circuits.

For $G$-COUNT$_\varepsilon$ one might expect that we would require instead $\text{NP}^{\text{QMA}}$, as QMA allows one to check whether two quantum circuits are approximately equal. However, this only works as a promise problem where either the circuits have to be closer than a certain bound $\varepsilon$ \emph{or} more different then some other bound $\varepsilon'$ and these bounds need to be `far apart'. However, in this case where we are generating candidate circuits non-deterministically, we have no such promise. Instead, we wish to solve the non-promise problem of whether given circuits $U$ and $V$ are $\varepsilon$-close. The contravariant form is then determining whether there exists a normalised state $\ket{\psi}$ and global phase $\alpha$ such that $\lVert (U-e^{i\alpha}V)\ket{\psi}\rVert_2>\varepsilon$. If $\ket{\psi}$ is efficiently representable as a tensor, this calculation of the norm can be represented as a tensor contraction, and hence is in $P^{\#P}$. 
The overall problem would then be in $\text{NP}^{\#P}$, since we are non-deterministically choosing a candidate state $\ket{\psi}$. An upper bound to the hardness of $G$-COUNT$_\varepsilon$ would then be $\text{NP}^{\text{NP}^{\#P}}$.
However, we do not have such a promise that $\ket{\psi}$ is efficiently preparable. Whether $\ket{\psi}$ can always be chosen in such a manner is in fact exactly the question of whether QMA = QCMA~\cite{wocjan2003two}.
We have not managed to find any non-conditional upper bounds to $G$-COUNT$_\varepsilon$ and leave this for future work.

\section{Conclusion and outlook}
We have shown that, as has long been suspected, many relevant problems in quantum circuit optimisation are indeed hard. In particular, the following problems are all NP-hard:
\begin{itemize}
	\item Optimising $T$-count or $T$-depth.
	\item Optimisation of $T$-count or $T$-depth of unitaries only matching the target unitary within some error-bound on the norm.
	\item Optimising Toffoli count or depth of classical reversible circuits.
	\item For any non-Clifford gate $G$, optimising the $G$-count or $G$-depth of Clifford+$G$ unitaries within some error bound on a target unitary.
	\item Optimising the number of CNOT gates in a Clifford+$T$ circuit.
	\item Optimising the number of Hadamard gates in a Clifford+$T$ circuit.
\end{itemize}
 A number of open questions remain:
\begin{itemize}
	\item What is the exact hardness of T-COUNT? It would make sense for this to be a complete problem for $\text{NP}^{\text{NQP}}$, but there are some problems with NQP requiring exact equality, while Clifford+T is only approximately universal. Even with the natural restriction of NQP to unitaries over the ``Clifford+$T$ domain'' $\mathbb{D}[\omega]$ for which Clifford+$T$ is exactly universal \cite{Giles2013a}, the proof of completeness for Boolean circuits does not extend trivially.
	\item Is it still hard to approximate the optimal T-count within a certain small error? Does it matter whether this error bound is additive or multiplicative?
	\item For Clifford+T we know that exact optimisation and approximate optimisation are both NP-hard, but for general non-Clifford gates $G$ we only know that approximate optimisation of Clifford+$G$ circuits is NP-hard (for small enough error bounds). Is exact optimisation of the $G$-count of Clifford+$G$ circuits NP-hard? 
	\item Optimising circuits consisting of only CNOT gates is a well-studied problem~\cite{markov2008optimal} and is closely related to optimising steps in a Gaussian elimination of a linear system. It seems likely that this problem is also at least NP-hard, but our methods do not suffice to prove that.
\end{itemize}

\medskip
\noindent\textbf{Acknowledgments}: The authors wish to thank Robin Kothari for pointing out that our original upper bound to the T-count problem of $\text{NP}^{\text{NP}^{\#\text{P}}}$ can be improved to $\text{NP}^{\text{NQP}}$, and Tuomas Laakkonen for suggesting that our hardness argument might also apply to Toffoli gate optimisation. MA acknowledges support from the Canada Research Chair program.

\bibliographystyle{plain}
\bibliography{bibliography}

\begin{thebibliography}{10}

\bibitem{aaronsongottesman2004}
Scott Aaronson and Daniel Gottesman.
\newblock Improved simulation of stabilizer circuits.
\newblock {\em Physical Review A}, 70(5):052328, 2004.

\bibitem{amy2017}
Matthew Amy, Jianxin Chen, and Neil J.~Ross.
\newblock {A Finite Presentation of CNOT-Dihedral Operators}.
\newblock In {\em Proceedings of the 14th International Conference on Quantum
  Physics and Logic}, QPL '17, pages 84--97, 2017.

\bibitem{amy2014polynomial}
Matthew Amy, Dmitri Maslov, and Michele Mosca.
\newblock {Polynomial-time T-depth optimization of Clifford+ T circuits via
  matroid partitioning}.
\newblock {\em IEEE Transactions on Computer-Aided Design of Integrated
  Circuits and Systems}, 33(10):1476--1489, 2014.

\bibitem{amy2016t}
Matthew Amy and Michele Mosca.
\newblock {T-count optimization and Reed-Muller codes}.
\newblock {\em Transactions on Information Theory}, 2019.

\bibitem{bauer2020quantum}
Bela Bauer, Sergey Bravyi, Mario Motta, and Garnet Kin-Lic Chan.
\newblock Quantum algorithms for quantum chemistry and quantum materials
  science.
\newblock {\em Chemical Reviews}, 120(22):12685--12717, 2020.

\bibitem{bookatz2012qmacomplete}
Adam~D Bookatz.
\newblock {QMA}-complete problems.
\newblock Preprint, 2012.

\bibitem{buchfuhrer2011}
David Buchfuhrer and Christopher Umans.
\newblock The complexity of boolean formula minimization.
\newblock {\em Journal of Computer and System Sciences}, 77(1):142--153, 2011.

\bibitem{choi2019tutorial}
Jaeho Choi and Joongheon Kim.
\newblock A tutorial on quantum approximate optimization algorithm (qaoa):
  Fundamentals and applications.
\newblock In {\em 2019 international conference on information and
  communication technology convergence (ICTC)}, pages 138--142. IEEE, 2019.

\bibitem{deBeaudrapN2020treducspidernest}
Niel de~Beaudrap, Xiaoning Bian, and Quanlong Wang.
\newblock {Fast and Effective Techniques for T-Count Reduction via Spider Nest
  Identities}.
\newblock In Steven~T. Flammia, editor, {\em 15th Conference on the Theory of
  Quantum Computation, Communication and Cryptography (TQC 2020)}, volume 158
  of {\em Leibniz International Proceedings in Informatics (LIPIcs)}, pages
  11:1--11:23, Dagstuhl, Germany, 2020. Schloss Dagstuhl--Leibniz-Zentrum
  f{\"u}r Informatik.

\bibitem{fowler2012surface}
Austin~G Fowler, Matteo Mariantoni, John~M Martinis, and Andrew~N Cleland.
\newblock Surface codes: {Towards} practical large-scale quantum computation.
\newblock {\em Physical Review A}, 86(3):032324, 2012.

\bibitem{gmc17}
H{\'e}ctor~J. Garc{\'i}a, Igor~L. Markov, and Andrew~W. Cross.
\newblock On the geometry of stabilizer states.
\newblock {\em Quantum Information and Computation}, 14:683--720, 2014.

\bibitem{gheorghiu2022quasi}
Vlad Gheorghiu, Michele Mosca, and Priyanka Mukhopadhyay.
\newblock A (quasi-) polynomial time heuristic algorithm for synthesizing
  t-depth optimal circuits.
\newblock {\em npj Quantum Information}, 8(1):110, 2022.

\bibitem{gheorghiu2022t}
Vlad Gheorghiu, Michele Mosca, and Priyanka Mukhopadhyay.
\newblock T-count and t-depth of any multi-qubit unitary.
\newblock {\em npj Quantum Information}, 8(1):141, 2022.

\bibitem{Gidney2021howtofactorbit}
Craig Gidney and Martin Eker{\aa{}}.
\newblock How to factor 2048 bit {RSA} integers in 8 hours using 20 million
  noisy qubits.
\newblock {\em {Quantum}}, 5:433, April 2021.

\bibitem{Gidney2019efficientmagicstate}
Craig Gidney and Austin~G. Fowler.
\newblock Efficient magic state factories with a catalyzed {$|CCZ\rangle$} to
  {$2|T\rangle$} transformation.
\newblock {\em {Quantum}}, 3:135, 4 2019.

\bibitem{Giles2013a}
Brett Giles and Peter Selinger.
\newblock {Exact synthesis of multiqubit Clifford+ T circuits}.
\newblock {\em Physical Review A}, 87(3):032332, 2013.

\bibitem{grover1996fast}
Lov~K Grover.
\newblock A fast quantum mechanical algorithm for database search.
\newblock In {\em Proceedings of the twenty-eighth annual ACM symposium on
  Theory of computing}, pages 212--219, 1996.

\bibitem{heyfron2018efficient}
Luke~E Heyfron and Earl~T Campbell.
\newblock {An efficient quantum compiler that reduces T count}.
\newblock {\em Quantum Science and Technology}, 4(015004), 2018.

\bibitem{kissinger2019tcount}
Aleks Kissinger and John van~de Wetering.
\newblock {Reducing the number of non-Clifford gates in quantum circuits}.
\newblock {\em Physical Review A}, 102:022406, 8 2020.

\bibitem{Litinski2019gameofsurfacecodes}
Daniel Litinski.
\newblock A {G}ame of {S}urface {C}odes: {L}arge-{S}cale {Q}uantum {C}omputing
  with {L}attice {S}urgery.
\newblock {\em {Quantum}}, 3:128, 3 2019.

\bibitem{Litinski2019magicstate}
Daniel Litinski.
\newblock Magic {S}tate {D}istillation: {N}ot as {C}ostly as {Y}ou {T}hink.
\newblock {\em {Quantum}}, 3:205, December 2019.

\bibitem{markov2008optimal}
Ketan Patel, Igor Markov, , and John Hayes.
\newblock Optimal synthesis of linear reversible circuits.
\newblock {\em Quantum Information and Computation}, 8(3\&4):0282--0294, 2008.

\bibitem{ruiz2024quantum}
Francisco~JR Ruiz, Tuomas Laakkonen, Johannes Bausch, Matej Balog, Mohammadamin
  Barekatain, Francisco~JH Heras, Alexander Novikov, Nathan Fitzpatrick,
  Bernardino Romera-Paredes, John van~de Wetering, et~al.
\newblock Quantum circuit optimization with alphatensor.
\newblock {\em arXiv preprint arXiv:2402.14396}, 2024.

\bibitem{shor1994algorithms}
Peter~W Shor.
\newblock Algorithms for quantum computation: discrete logarithms and
  factoring.
\newblock In {\em Proceedings 35th annual symposium on foundations of computer
  science}, pages 124--134. Ieee, 1994.

\bibitem{shor1995scheme}
Peter~W Shor.
\newblock Scheme for reducing decoherence in quantum computer memory.
\newblock {\em Physical review A}, 52(4):R2493, 1995.

\bibitem{tanaka2010exact}
Yu~Tanaka.
\newblock Exact non-identity check is nqp-complete.
\newblock {\em International Journal of Quantum Information}, 8(05):807--819,
  2010.

\bibitem{tilly2022variational}
Jules Tilly, Hongxiang Chen, Shuxiang Cao, Dario Picozzi, Kanav Setia, Ying Li,
  Edward Grant, Leonard Wossnig, Ivan Rungger, George~H Booth, et~al.
\newblock The variational quantum eigensolver: a review of methods and best
  practices.
\newblock {\em Physics Reports}, 986:1--128, 2022.

\bibitem{vandewetering2024optimal}
John van~de Wetering, Richie Yeung, Tuomas Laakkonen, and Aleks Kissinger.
\newblock {Optimal compilation of parametrised quantum circuits}.
\newblock {\em arXiv preprint arXiv:2401.12877}, 2024.

\bibitem{wocjan2003two}
Pawel Wocjan, Dominik Janzing, and Thomas Beth.
\newblock Two qcma-complete problems.
\newblock {\em arXiv preprint quant-ph/0305090}, 2003.

\bibitem{zhang2019optimizing}
Fang Zhang and Jianxin Chen.
\newblock Optimizing {T} gates in {Clifford+T} circuit as $\pi/4$ rotations
  around {Paulis}.
\newblock Preprint, 2019.

\end{thebibliography}

\end{document}